\documentclass[12pt]{article}
\usepackage[left=35mm,right=35mm,top=35mm,bottom=35mm]{geometry}
\usepackage{amsmath}
\usepackage{amsthm}
\usepackage{amssymb}
\usepackage{amsfonts}
\usepackage{mathrsfs}
 
\DeclareMathOperator*{\slim}{s-lim}
\DeclareMathOperator*{\supp}{supp}

\newcommand{\Eqn}[1]{&\hspace{-0.5em}#1\hspace{-0.5em}&}
\makeatletter

\@addtoreset{equation}{section}
\makeatother
\newtheorem{The}{Theorem}[section]
\newtheorem{Ass}[The]{Assumption}

\newtheorem{Prop}[The]{Proposition}
\newtheorem{Cor}[The]{Corollary}

\title{Propagation property and its application\\to inverse scattering for fractional powers\\of the negative Laplacian}
\author{Atsuhide ISHIDA\\
\\
Department of Liberal Arts, Faculty of Engineering,
 \\Tokyo University of Science\\
\normalsize 3-1 Niijuku, 6-chome, Katsushika-ku,Tokyo 125-8585, Japan\\
\normalsize E-mail: aishida@rs.tus.ac.jp
}
\date{}

\begin{document}
\maketitle

\begin{abstract}
Enss (1983) proved a propagation estimate for the usual free Schr\"odinger operator that turned out later to be very useful for inverse scattering in the work of Enss--Weder (1995). Since then, this method has been called the Enss--Weder time-dependent method. We study the same type of propagation estimate for the fractional powers of the negative Laplacian and, as with the Enss--Weder method, we apply our estimate to inverse scattering. We find that the high-velocity limit of the scattering operator uniquely determines the short-range interactions.
\end{abstract}

\quad\textit{Keywords}: scattering theory, inverse problem, fractional Laplacian\par
\quad\textit{MSC}2010: 81Q10, 81U05, 81U40

\section{Introduction}
For $1/2\leqslant\rho\leqslant1$, the fractional powers of the negative Laplacian as self-adjoint operator acting on $L^2(\mathbb{R}^n)$ is defined by the Fourier multiplier with symbol
\begin{equation}
\omega_\rho(\xi)=|\xi|^{2\rho}/(2\rho).
\end{equation}
We denote this operator by
\begin{equation}
H_{0,\rho}=\omega_\rho(D_x),\label{free}
\end{equation}
where $D_x=-i\nabla_x=-i(\partial_{x_1},\dots,\partial_{x_n})$. More specifically, we can represent $H_{0,\rho}$ by the Fourier integral operator
\begin{eqnarray}
H_{0,\rho}\phi(x)\Eqn{=}(\mathscr{F}^*\omega_\rho(\xi)\mathscr{F}\phi)(x)\nonumber\\
\Eqn{=}\int_{\mathbb{R}^n}e^{ix\cdot\xi}\omega_\rho(\xi)(\mathscr{F}\phi)(\xi)d\xi/(2\pi)^{n/2}\nonumber\\
\Eqn{=}\int_{\mathbb{R}^{2n}}e^{i(x-y)\cdot\xi}\omega_\rho(\xi)\phi(y)dyd\xi/(2\pi)^n
\end{eqnarray}
for $\phi\in\mathscr{D}(H_{0,\rho})=H^{2\rho}(\mathbb{R}^n)$, which is the Sobolev space of order $2\rho$. In particular, if $\rho=1$, then $H_{0,1}$ is the free Schr\"odinger operator $\omega_1(D_x)=-\Delta_x/2=-\sum_{j=1}^n\partial_{x_j}^2/2$. If $\rho=1/2$, then $H_{0,1/2}$ is the massless relativistic Schr\"odinger operator $\omega_{1/2}(D_x)=\sqrt{-\Delta_x}$.\par
In Section 2, we prove the following Enss-type propagation estimate for $e^{-itH_{0,\rho}}$. Throughout this paper, $F(\cdots)$ is the usual characteristic function of the set $\{\cdots\}$. We denote the smooth characteristic function $\chi\in C^\infty(\mathbb{R}^n)$ by
\begin{equation}
\chi(x)=
\begin{cases}
\ 1\quad|x|\geqslant2,\\
\ 0\quad|x|\leqslant1.
\end{cases}
\end{equation}

\begin{The}\label{the1}
Let $f\in C_0^\infty(\mathbb{R}^n)$ with $\supp f\subset\{ \xi\in\mathbb{R}^n\bigm||\xi|\leqslant\eta\}$ for some given $\eta>0$. Choose $v\in\mathbb{R}^n$ such that $|v|>\eta$ and
\begin{equation}
\left\{
\begin{array}{cc}
16n(1-\rho)(|v|-\eta)^{2\rho-2}\eta\leqslant|v|^{2\rho-1} & 1/2\leqslant\rho<1,\\
8\eta\leqslant|v| & \rho=1.
\end{array}
\right.\label{size_v}
\end{equation}
For $t\in\mathbb{R}$ and $N\in\mathbb{N}$, the following estimate holds.
\begin{eqnarray}
\lefteqn{\left\|\chi\left(\frac{x-(\nabla_\xi\omega_\rho)(v)t}{|v|^{2\rho-1}|t|/4}\right)e^{-itH_{0,\rho}}f(D_x-v)F\left(|x|\leqslant\frac{|v|^{2\rho-1}|t|}{16}\right)\right\|}\nonumber\\
\Eqn{}\hspace{75mm}\leqslant C_N(1+|v|^{2\rho-1}|t|)^{-N},\label{fractional}
\end{eqnarray}
where $\|\cdot\|$ stands for the operator norm on $L^2(\mathbb{R}^n)$, and the constant $C_N>0$ also depends on the dimension $n$ and the shape of $f$.
\end{The}
Enss \cite{En} proved the following estimate for the free Schr\"odinger operator
\begin{equation}
\left\|F\left(|x-vt|\geqslant\frac{|v||t|}{4}\right)e^{-itD_x^2/2}f(D_x-v)F\left(|x|\leqslant\frac{|v||t|}{16}\right)\right\|\leqslant C_N(1+|v||t|)^{-N}.\label{enss}
\end{equation}
This estimate was proved not only for the spheres but more generally for the measurable subsets of $\mathbb{R}^n$ (see Proposition 2.10 in Enss \cite{En}). Before considering Theorem \ref{the1} further, we discuss the meaning of the estimate \eqref{enss}. From the perspective of classical mechanics, $D_x$ represents the momentum or equivalently the velocity of the particle of unit mass. On the left-hand side of \eqref{enss}, $D_x$ is localized to the neighborhood of $v$ by the cut-off function $f$. Therefore, along the time evolution of the propagator $e^{-itD_x^2/2}$, the position of the particle behaves according to
\begin{equation}
x\sim D_xt\sim vt.
\end{equation}
Because the behavior of the points on the sphere is the same, the center of the sphere moves toward $vt$ from the origin
\begin{equation}
\left\{ x\in\mathbb{R}^n\biggm||x|\leqslant\frac{|v||t|}{16}\right\}\sim\left\{ x\in\mathbb{R}^n\biggm||x-vt|\leqslant\frac{|v||t|}{16}\right\}.\label{sphere}
\end{equation}
We extract an interpretation of the estimate \eqref{enss} from these observations. The behavior of the sphere \eqref{sphere} makes the characteristic functions on both sides of \eqref{enss} disjoint. Thus, this gives rise to the decay associated with time and velocity. Theorem \ref{the1} is the fractional Laplacian version of \eqref{enss}. From $(\nabla_\xi\omega_\rho)(v)=|v|^{2\rho-2}v$, the case where $\rho=1$ in \eqref{fractional} is essentially equivalent to \eqref{enss}. Conversely, if $\rho=1/2$ in \eqref{fractional}, the decay on the right-hand side does not involve $|v|$. However, this does not conflict with the physical meaning. In the case where $\rho=1/2$, the system is relativistic. In this system, the particle does not have a mass, and its velocity is the speed of light, which is normalized to $1$. Therefore, the decay cannot include the velocity $v$. \par
Spectral analysis for the relativistic Schr\"odinger operator was initiated by Weder \cite{We1}, following which Umeda \cite{Um1,Um2} studied the resolvent estimate and mapping properties associated with the Sobolev spaces. Wei \cite{Wei} studied the generalized eigenfunctions. Weder \cite{We2} analyzed the spectral properties of the fractional Laplacian for the massive case, and Watanabe \cite{Wa} studied the Kato-smoothness. Giere \cite{Gi} investigated the scattering theory and proved the asymptotic completeness of the wave operators for short-range perturbations. Recently, Kitada \cite{Ki1,Ki2} constructed the long-range theory.\par
In Section 3, we assume that the dimension of the space satisfies $n\geqslant2$. As an application of Theorem \ref{the1}, we consider a multidimensional inverse scattering. The high-velocity limit of the scattering operator uniquely determines the interaction potentials that satisfy the short-range condition below by using the Enss--Weder time-dependent method (Enss--Weder \cite{EnWe}).

\begin{Ass}\label{ass1}
$V\in C^1(\mathbb{R}^n)$ is real-valued and for $\gamma>1$, satisfies
\begin{equation}
|\partial_x^\beta V(x)|\leqslant C_\beta \langle x\rangle^{-\gamma-|\beta|},\quad |\beta|\leqslant1,\label{potential}
\end{equation}
where the bracket of $x$ has the usual definition $\langle x\rangle=\sqrt{1+|x|^2}$.
\end{Ass}
For the full Hamiltonian $H_\rho=H_{0,\rho}+V$, where $V$ belongs to the class stated above, the existence of the wave operators
\begin{equation}
W_\rho^\pm=\slim_{t\rightarrow \pm \infty}e^{itH_\rho}e^{-itH_{0,\rho}}\label{wave_operators}
\end{equation}
and their asymptotic completeness have already been proved in Kitada \cite{Ki1,Ki2}. Thus, we can define the scattering operator $S_\rho=S_\rho(V)$ by
\begin{equation}
S_\rho=(W_\rho^+)^*W_\rho^-.
\end{equation}
Under these situations, the following uniqueness theorem can be proved.

\begin{The}\label{the2}
Let $V_1$ and $V_2$ be interaction potentials which satisfy Assumption \ref{ass1}. If $S_\rho(V_1)=S_\rho(V_2)$, then $V_1=V_2$ holds for $1/2<\rho\leqslant1$.
\end{The}
We note that $\rho=1/2$ is excluded in this theorem. As mentioned before, if $\rho=1/2$, the system is relativistic and the speed of light is always equal to 1, that is, $|v|\equiv1$. The Enss--Weder time-dependent method is also called the high-velocity method. As the name suggests, deriving the uniqueness of the interaction potentials requires the limit of $|v|$. Therefore, this method does not combine well with relativistic phenomena (see also Jung \cite{Ju}).\par
In Enss--Weder \cite{EnWe}, the estimate \eqref{enss} was demonstrated to be very useful for inverse scattering and the Enss--Weder time-dependent method was developed. Since then, the uniqueness of the interaction potentials for various quantum systems has been studied by many authors (Weder \cite{We3}, Jung \cite{Ju}, Nicoleau \cite{Ni1,Ni2,Ni3}, Adachi--Maehara \cite{AdMa}, Adachi--Kamada--Kazuno--Toratani \cite{AdKaKaTo}, Valencia--Weder \cite{VaWe}, Adachi--Fujiwara--Ishida \cite{AdFuIs}, and Ishida \cite{Is}). This paper is motivated by their results. In particular, Enss--Weder \cite{EnWe} first proved the uniqueness of the potentials for $\rho=1$ by applying \eqref{enss}. Jung \cite{Ju} treated $\rho=1/2$ using a different approach. Naturally, we cannot consider the limit of the velocity in this case. However, Jung \cite{Ju} obtained the uniqueness without using an estimate of the type \eqref{fractional}. Thus, Theorem \ref{the2} represents an interpolation between the results of Enss--Weder \cite{EnWe} and Jung \cite{Ju}.

\section{Propagation Property}
In this section, we prove Theorem \ref{the1}. Regarding estimate \eqref{enss}, the idea of Enss \cite{En} is very simple and understandable. The Galilean transformation in the direction of $v$ enables a reduction to a static system, and iterations of the integration by parts, by taking the points of stationary phase into account, leads to \eqref{enss}. However, in our case, these ingredients do not work well because of the fractional powers. Instead, our main strategy is the asymptotic expansion of the symbolic calculus of pseudo-differential theory.\par

Here, we recall several basics of the calculus of pseudo-differential operators. They are recounted from standard textbooks. For $m\in \mathbb{R}$, let $S^m_{1,0}$ be the H\"ormander symbol class, that is, we say $p\in S^m_{1,0}$ if and only if $p\in C^\infty(\mathbb{R}^n_x\times\mathbb{R}^n_\xi)$ and, for any multi-indices $\beta$ and $\beta'$,
\begin{equation}
|\partial_x^{\beta'}\partial_\xi^\beta p(x,\xi)|\leqslant C_{\beta\beta'}\langle \xi\rangle^{m-|\beta|}
\end{equation}
are satisfied. Then, the pseudo-differential operator $p(x,D_x)$ with symbol $p\in S^m_{1,0}$ is defined by
\begin{equation}
p(x,D_x)\phi(x)=\int_{\mathbb{R}^n}e^{ix\cdot\xi}p(x,\xi)(\mathscr{F}\phi)(\xi)d\xi/(2\pi)^{n/2}
\end{equation}
for $\phi\in\mathscr{S}(\mathbb{R}^n)$ which is the Schwartz functional space. When $p\in S^m_{1,0}$, we denote the semi-norm $|p|_{m,k}$ by
\begin{equation}
|p|_{m,k}=\sup_{x,\xi\in\mathbb{R}^n}\sum_{|\beta|+|\beta'|\leqslant k}\langle \xi\rangle^{-m+|\beta|}|\partial_x^{\beta'}\partial_\xi^\beta p(x,\xi)|.
\end{equation}
If $p_1\in S^{m_1}_{1,0}$ and $p_2\in S^{m_2}_{1,0}$, then the symbol of the product $p_1p_2=q\in S^{m_1+m_2}_{1,0}$ has the following asymptotic expansion
\begin{equation}
q(x,\xi)=\sum_{|\beta|\leqslant N-1}\partial_\xi^\beta p_1(x,\xi)\times(-i\partial_x)^\beta p_2(x,\xi)/\beta!+r_N(x,\xi),\label{product_formula}
\end{equation} 
where the remainder $r_N$ satisfies $r_N\in S^{m_1+m_2-N}_{1,0}$ and
\begin{eqnarray}
\lefteqn{|\partial_x^{\beta'}\partial_\xi^\beta r_N(x,\xi)|}\nonumber\\
\Eqn{\leqslant}C_{\beta\beta'N}\sum_{|\alpha|=N}|\partial_\xi^{\alpha} p_1|_{m_1-N,M+|\beta|+|\beta'|}|\partial_x^{\alpha} p_2|_{m_2,M+|\beta|+|\beta'|}\langle\xi\rangle^{m_1+m_2-N-|\beta|}.
\end{eqnarray}
for some $M\in\mathbb{N}$ (Chapter 8 in Wong \cite{Wo}). Moreover, by the $L^2$-boundedness theorem, if $m_1+m_2-N\leqslant0$, then there exists $K\in\mathbb{N}$ such that the operator-norm of $r_N$ is estimated by
\begin{eqnarray}
\lefteqn{\|r_N(x,D_x)\|\leqslant C_N|r_N|_{m_1+m_2-N,K}}\nonumber\\
\Eqn{}\leqslant C_N\sup_{x,\xi\in\mathbb{R}^n}\sum_{|\beta|+|\beta'|\leqslant K}\langle \xi\rangle^{-m_1-m_2+N+|\beta|}|\partial_x^{\beta'}\partial_\xi^\beta r_N(x,\xi)|\nonumber\\
\Eqn{}\leqslant C_N\sup_{x,\xi\in\mathbb{R}^n}\sum_{\substack{|\beta|+|\beta'|\leqslant K\\ |\alpha|=N}}|\partial_\xi^{\alpha} p_1|_{m_1-N,M+|\beta|+|\beta'|}|\partial_x^{\alpha} p_2|_{m_2,M+|\beta|+|\beta'|}\label{bounds}
\end{eqnarray}
(Theorem 3.36, Lemma 3.37--3.39 and Remark 3.40 in Abels in \cite{Ab})

\begin{proof}[Proof of Theorem \ref{the1}]
The left-hand side of \eqref{fractional} is bounded uniformly in $t$ and $v$. Therefore, it is sufficient to prove
\begin{eqnarray}
\lefteqn{\left\|\chi\left(\frac{x-(\nabla_\xi\omega_\rho)(v)t}{|v|^{2\rho-1}|t|/4}\right)e^{-itH_{0,\rho}}f(D_x-v)F\left(|x|\leqslant\frac{|v|^{2\rho-1}|t|}{16}\right)\right\|}\nonumber\\
\Eqn{}\hspace{80mm}\leqslant C_N(|v|^{2\rho-1}|t|)^{-N}\label{fractional2}
\end{eqnarray}
for $|v|^{2\rho-1}|t|\geqslant1$. By using the unitary translations, we have the following relations
\begin{gather}
e^{iv\cdot x}D_xe^{-iv\cdot x}=D_x-v,\label{eq1.1.1}\\
e^{it\omega_\rho(D_x+v)}xe^{-it\omega_\rho(D_x+v)}=x+(\nabla_\xi\omega_\rho)(D_x+v)t.\label{eq1.1.2}
\end{gather}
We thus compute that
\begin{eqnarray}
\lefteqn{\chi\left(\frac{x-(\nabla_\xi\omega_\rho)(v)t}{|v|^{2\rho-1}|t|/4}\right)e^{-itH_{0,\rho}}f(D_x-v)}\nonumber\\
\Eqn{=}e^{iv\cdot x}\chi\left(\frac{x-(\nabla_\xi\omega_\rho)(v)t}{|v|^{2\rho-1}|t|/4}\right)e^{-it\omega_\rho(D_x+v)}f(D_x)e^{-iv\cdot x}\nonumber\\
\Eqn{=}e^{iv\cdot x}e^{-it\omega_\rho(D_x+v)}\chi\left(\frac{x+(\nabla_\xi\omega_\rho)(D_x+v)t-(\nabla_\xi\omega_\rho)(v)t}{|v|^{2\rho-1}|t|/4}\right)f(D_x)e^{-iv\cdot x}.\quad\label{eq1.1.3}
\end{eqnarray}
The strategy of our proof is as follows. The momentum operator $D_x$ can move inside the compact region only because $f$ is compactly supported. Therefore, $(\nabla_\xi\omega_\rho)(D_x+v)$ and $(\nabla_\xi\omega_\rho)(v)$ almost cancel when $|v|$ is sufficiently large, and the function $\chi$ in \eqref{eq1.1.3} behaves as though
\begin{equation}
\chi\left(\frac{x+(\nabla_\xi\omega_\rho)(D_x+v)t-(\nabla_\xi\omega_\rho)(v)t}{|v|^{2\rho-1}|t|/4}\right)\sim\chi\left(\frac{x}{|v|^{2\rho-1}|t|/4}\right).
\end{equation}
We now justify this strategy. Because $|\xi|\leqslant\eta$ on the support of $f$, we have
\begin{equation}
|\xi+v|\geqslant|v|-|\xi|\geqslant|v|-\eta>0.
\end{equation}
This inequality implies
\begin{equation}
\chi\left(\frac{x+(\nabla_\xi\omega_\rho)(\xi+v)t-(\nabla_\xi\omega_\rho)(v)t}{|v|^{2\rho-1}|t|/4}\right)f(\xi)\in C^\infty(\mathbb{R}^n_x\times\mathbb{R}^n_\xi).
\end{equation}
Moreover, when $1/2\leqslant\rho<1$,
\begin{equation}
|(\nabla_\xi\omega_\rho)(\xi+v)-(\nabla_\xi\omega_\rho)(v)|\leqslant\int_0^1|(\nabla_\xi^2\omega_\rho)(v+\theta\xi)|d\theta|\xi|\label{eq1.1.a}
\end{equation}
and
\begin{equation}
|(\nabla_\xi^2\omega_\rho)(v+\theta\xi)|=\max_{1\leqslant j\leqslant n}\sum_{k=1}^n|(\partial_{\xi_j}\partial_{\xi_k}\omega_\rho)(v+\theta\xi)|\leqslant2n(1-\rho)(|v|-\eta)^{2\rho-2}.\label{eq1.1.b}
\end{equation}
hold for $|\xi|\leqslant\eta$, where $\nabla_\xi^2\omega_\rho$ denotes the Hessian matrix of $\omega_\rho$. In the case where $\rho=1$, it is clear that
\begin{equation}
|(\nabla_\xi\omega_1)(\xi+v)-(\nabla_\xi\omega_1)(v)|=|\xi|.\label{eq1.1.c}
\end{equation}
We thus obtain, for $1/2\leqslant\rho\leqslant1$ and $|v|$ which satisfies \eqref{size_v},
\begin{equation}
|(\nabla_\xi\omega_\rho)(\xi+v)-(\nabla_\xi\omega_\rho)(v)|\leqslant|v|^{2\rho-1}/8.\label{eq1.1.4}
\end{equation}
It follows from \eqref{eq1.1.4} that
\begin{eqnarray}
|x|\Eqn{\geqslant}|x+(\nabla_\xi\omega_\rho)(\xi+v)t-(\nabla_\xi\omega_\rho)(v)t|-|(\nabla_\xi\omega_\rho)(\xi+v)-(\nabla_\xi\omega_\rho)(v)||t|\nonumber\\
\Eqn{\geqslant}|v|^{2\rho-1}|t|/4-|v|^{2\rho-1}|t|/8=|v|^{2\rho-1}|t|/8,\label{eq1.1.5}
\end{eqnarray}
on the supports of $f$ and $\chi$. This means that
\begin{eqnarray}
\lefteqn{\chi\left(\frac{x+(\nabla_\xi\omega_\rho)(\xi+v)t-(\nabla_\xi\omega_\rho)(v)t}{|v|^{2\rho-1}|t|/4}\right)f(\xi)}\nonumber\\
\Eqn{=}\chi\left(\frac{x+(\nabla_\xi\omega_\rho)(\xi+v)t-(\nabla_\xi\omega_\rho)(v)t}{|v|^{2\rho-1}|t|/4}\right)f(\xi)\chi\left(\frac{x}{|v|^{2\rho-1}|t|/16}\right)\label{eq1.1.6}
\end{eqnarray}
because $\chi(x/(|v|^{2\rho-1}|t|/16))=1$ by \eqref{eq1.1.5}. However, in the pseudo-differential calculus, the product of the symbols is not equal to the symbol of the product. The additional asymptotic error terms arise. By the product formula \eqref{product_formula}, the symbol of \eqref{eq1.1.6} becomes
\begin{eqnarray}
\lefteqn{\sum_{|\beta|\leqslant N-1}\frac{1}{\beta!}\partial_\xi^\beta\left\{\chi\left(\frac{x+(\nabla_\xi\omega_\rho)(\xi+v)t-(\nabla_\xi\omega_\rho)(v)t}{|v|^{2\rho-1}|t|/4}\right)f(\xi)\right\}}\nonumber\\
\Eqn{}\hspace{35mm}\times(-i\partial_x)^\beta\chi\left(\frac{x}{|v|^{2\rho-1}|t|/16}\right)+R_N(t,x,\xi)
\end{eqnarray}
for any $N\in\mathbb{N}$. All terms with $|\beta|\leqslant N-1$ vanish due to another characteristic function
\begin{equation}
\left\{\partial_x^\beta\chi\left(\frac{x}{|v|^{2\rho-1}|t|/16}\right)\right\}F\left(|x|\leqslant\frac{|v|^{2\rho-1}|t|}{16}\right)=0.
\end{equation}
Next, we consider the remainder term $R_N$. Because $f$ is compactly supported,
\begin{equation}
\chi\left(\frac{x+(\nabla_\xi\omega_\rho)(\xi+v)t-(\nabla_\xi\omega_\rho)(v)t}{|v|^{2\rho-1}|t|/4}\right)f(\xi)\in S_{1,0}^{-\infty}=\bigcap_{-\infty<m<\infty}S_{1,0}^m
\end{equation}
holds. Clearly, $\chi(x/(|v|^{2\rho-1}|t|/16))\in S_{1,0}^0$ also holds. In particular, we see that
\begin{equation}
\left|\partial_\xi^\beta\chi\left(\frac{x+(\nabla_\xi\omega_\rho)(\xi+v)t-(\nabla_\xi\omega_\rho)(v)t}{|v|^{2\rho-1}|t|/4}\right)\right|\leqslant C_\beta|v|^{-(2\rho-1)}\leqslant C_\beta
\end{equation}
for all $\beta$ with $|\beta|\geqslant1$. Here, $C_\beta>0$ is independent of $t$ and $v$. Therefore, it is sufficient to focus only on the derivative at $x$. By the estimate of the remainder \eqref{bounds}, there exists $N'\in\mathbb{N}$ such that
\begin{eqnarray}
\|R_N(t,x,D_x)\|\Eqn{\leqslant}C_N\sum_{0\leqslant j\leqslant N'}(|v|^{2\rho-1}|t|)^{-j}\times\sum_{N\leqslant j\leqslant N+N'}(|v|^{2\rho-1}|t|)^{-j}\nonumber\\
\Eqn{\leqslant}C_N(|v|^{2\rho-1}|t|)^{-N}\label{eq1.1.7}
\end{eqnarray}
because $|v|^{2\rho-1}|t|\geqslant1$. This completes the proof.
\end{proof}

\section{Uniqueness of Interactions}
To apply the Enss--Weder time-dependent method, we have to assume that $n\geqslant2$ and that $\rho>1/2$ from here on. The following Radon transformation-type reconstruction formula enables Theorem \ref{the2} to be proved. We devote ourselves to proving Theorem \ref{the3} in this section. Contrary to Enss--Weder \cite{EnWe}, the key calculation in our proof is the pseudo-differential asymptotic expansion as in Theorem \ref{the1}.

\begin{The}\label{the3}
Let $v\in\mathbb{R}^n$ be given and let $\hat{v}=v/|v|$. Suppose that $\eta>0$, and that $\Phi_0,\Psi_0\in L^2(\mathbb{R}^n)$ such that $\mathscr{F}\Phi_0,\mathscr{F}\Psi_0\in C_0^\infty(\mathbb{R}^n)$ with $\supp\mathscr{F}\Phi_0,\supp\mathscr{F}\Psi_0\subset\{\xi\in\mathbb{R}^n\bigm||\xi|\leqslant\eta\}$. Let $\Phi_v=e^{iv\cdot x}\Phi_0,\Psi_v=e^{iv\cdot x}\Psi_0$. Then
\begin{equation}
|v|^{2\rho-1}(i(S_\rho-1)\Phi_v,\Psi_v)=\int_{-\infty}^\infty(V(x+\hat{v}t)\Phi_0,\Psi_0)dt+o(1)\label{reconstruction}
\end{equation}
holds as $|v|\rightarrow \infty$ for any $V$ which satisfies Assumption \ref{ass1}, where $(\cdot,\cdot)$ is the scalar product of $L^2(\mathbb{R}^n)$.
\end{The}

We first prepare the propagation estimate of the following integral form. In the proof of this proposition, we can see that Theorem \ref{the1} plays an important role. While $\|\cdot\|$ also indicates the norm in $L^2(\mathbb{R}^n)$, for simplicity, we do not distinguish between the notations for the usual $L^2$-norm and its operator norm in this paper.

\begin{Prop}\label{prop}
Let $v$ and $\Phi_v$ be as in Theorem \ref{the3}. Then
\begin{equation}
\int_{-\infty}^\infty\|V(x)e^{-itH_{0,\rho}}\Phi_v\|dt=O(|v|^{1-2\rho})
\end{equation}
holds as $|v|\rightarrow\infty$ for any $V$ which satisfies Assumption \ref{ass1}.
\end{Prop}

\begin{proof}
The original idea of this proof is given in Lemma 2.2 of Enss--Weder \cite{EnWe}. We extend it to the case of the fractional powers of the negative Laplacian. Choose $f\in C_0^\infty(\mathbb{R}^n)$ such that $\mathscr{F}\Phi_0=f\mathscr{F}\Phi_0$ and $\supp f\subset\{\xi\in\mathbb{R}^n\bigm||\xi|\leqslant\eta\}$. Then the relation
\begin{equation}
\Phi_v=e^{iv\cdot x}\mathscr{F}^*f(\xi)\mathscr{F}\Phi_0=e^{iv\cdot x}f(D_x)\Phi_0=f(D_x-v)\Phi_v
\end{equation}
follows. We compute
\begin{equation}
\|V(x)e^{-itH_{0,\rho}}\Phi_v\|=\|V(x)e^{-itH_{0,\rho}}f(D_x-v)\Phi_v\|\leqslant I_1+I_2,
\end{equation}
where $I_1$ and $I_2$ are given by
\begin{eqnarray}
I_1\Eqn{=}\left\|V(x)\left\{1-\chi\left(\frac{x-(\nabla_\xi\omega_\rho)(v)t}{|v|^{2\rho-1}|t|/4}\right)\right\}e^{-itH_{0,\rho}}f(D_x-v)\Phi_v\right\|,\\
I_2\Eqn{=}\left\|V(x)\chi\left(\frac{x-(\nabla_\xi\omega_\rho)(v)t}{|v|^{2\rho-1}|t|/4}\right)e^{-itH_{0,\rho}}f(D_x-v)\Phi_v\right\|.
\end{eqnarray}
When $|x-(\nabla_\xi\omega_\rho)(v)t|\leqslant|v|^{2\rho-1}|t|/2$ holds, we have
\begin{equation}
|x|\geqslant|(\nabla_\xi\omega_\rho)(v)t|-|x-(\nabla_\xi\omega_\rho)(v)t|\geqslant|v|^{2\rho-1}|t|/2.\label{eq3.2.2}
\end{equation}
By virtue of the decay condition on $V$ in \eqref{potential} and inequality \eqref{eq3.2.2}, $I_1$ can be estimated as follows
\begin{equation}
\int_{-\infty}^{\infty}I_1dt\leqslant C\int_0^\infty\langle|v|^{2\rho-1}t\rangle^{-\gamma}dt=C|v|^{1-2\rho}\int_0^\infty\langle\tau\rangle^{-\gamma}d\tau=O(|v|^{1-2\rho})\label{eq3.2.3}
\end{equation}
because $\gamma>1$, where we changed the integral variable by $\tau=|v|^{2\rho-1}t$. We next estimate $I_2$. By inserting
\begin{equation}
F\left(|x|\leqslant\frac{|v|^{2\rho-1}|t|}{16}\right)+F\left(|x|>\frac{|v|^{2\rho-1}|t|}{16}\right)=1
\end{equation}
between $f(D_x-v)$ and $\Phi_v$, $I_2$ is estimated so that $I_2\leqslant I_{2,1}+I_{2,2}$ where $I_{2,1}$ and $I_{2,2}$ are given by
\begin{eqnarray}
I_{2,1}\Eqn{=}C\left\|\chi\left(\frac{x-(\nabla_\xi\omega_\rho)(v)t}{|v|^{2\rho-1}|t|/4}\right)e^{-itH_{0,\rho}}f(D_x-v)F\left(|x|\leqslant\frac{|v|^{2\rho-1}|t|}{16}\right)\right\|,\quad\\
I_{2,2}\Eqn{=}C\left\|F\left(|x|>\frac{|v|^{2\rho-1}|t|}{16}\right)\Phi_0\right\|.
\end{eqnarray}
By applying Theorem \ref{the1} to $I_{2.1}$ with $N=2$
\begin{equation}
\int_{-\infty}^{\infty}I_{2.1}dt\leqslant C\int_0^\infty\langle|v|^{2\rho-1}t\rangle^{-2}dt=C|v|^{1-2\rho}\int_0^\infty\langle\tau\rangle^{-2}d\tau=O(|v|^{1-2\rho})\label{eq3.2.4}
\end{equation}
is obtained. $I_{2,2}$ also provides the same estimate of \eqref{eq3.2.4}. Indeed, $I_{2,2}$ satisfies
\begin{equation}
I_{2,2}\leqslant C\left\|F\left(|x|>\frac{|v|^{2\rho-1}|t|}{16}\right)\langle x\rangle^{-2}\right\|\|\langle x\rangle^2\Phi_0\|\leqslant C\langle|v|^{2\rho-1}|t|\rangle^{-2}\label{eq3.2.5}
\end{equation}
because $\Phi_0\in\mathscr{S}(\mathbb{R}^n)$ by the assumption. Therefore, we obtain
\begin{equation}
\int_{-\infty}^{\infty}I_{2.2}dt=O(|v|^{1-2\rho})).\label{eq3.2.6}
\end{equation}
From \eqref{eq3.2.3}, \eqref{eq3.2.4}, and \eqref{eq3.2.6}, it follows that
\begin{equation}
\int_{-\infty}^\infty\|V(x)e^{-itH_{0,\rho}}\Phi_v\|dt\leqslant\int_{-\infty}^\infty(I_1+I_{2.1}+I_{2.2})dt=O(|v|^{1-2\rho}).
\end{equation}
\end{proof}

\begin{Cor}\label{cor}
Let $v$ and $\Phi_v$ be as in Theorem \ref{the3}. Then
\begin{equation}
\|(W_\rho^\pm-1)e^{-itH_{0,\rho}}\Phi_v\|=O(|v|^{1-2\rho})
\end{equation}
holds as $|v|\rightarrow\infty$ uniformly for $t\in\mathbb{R}$.
\end{Cor}

\begin{proof}
The proof is similar to that of Corollary 2.3 in Enss--Weder \cite{EnWe}, and therefore is sketched as follows. The difference between $W_\rho^\pm$ and $1$ can be represented by the following integral form
\begin{eqnarray}
\lefteqn{(W_\rho^\pm-1)e^{-itH_{0,\rho}}=\int_0^{\pm\infty}\partial_\tau e^{i\tau H_\rho}e^{-i\tau H_{0,\rho}}d\tau e^{-itH_{0,\rho}}}\nonumber\\
\Eqn{=}i\int_0^{\pm\infty}e^{i\tau H_\rho}V(x)e^{-i(\tau+t)H_{0,\rho}}d\tau=i\int_t^{\pm\infty}e^{i(\tau'-t)H_\rho}V(x)e^{-i\tau'H_{0,\rho}}d\tau'.
\end{eqnarray}
In the last equation, we changed the integral variable $\tau'=\tau+t$. By using Proposition \ref{prop}, we have
\begin{equation}
\|(W_\rho^\pm-1)e^{-itH_{0,\rho}}\Phi_v\|\leqslant\int_{-\infty}^\infty\|V(x)e^{-i\tau'H_{0,\rho}}\Phi_v\|d\tau'=O(|v|^{1-2\rho}).
\end{equation}
\end{proof}
We are ready to prove the reconstruction theorem.

\begin{proof}[Proof of Theorem \ref{the3}]
As in the proof of Corollary \ref{cor}, we represent the difference between $W^+$ and $W^-$ by the integral
\begin{equation}
W_\rho^+-W_\rho^-=\int_{-\infty}^\infty\partial_te^{itH_\rho}e^{-itH_{0,\rho}}dt=i\int_{-\infty}^\infty e^{itH_\rho}V(x)e^{-itH_{0,\rho}}dt.
\end{equation}
Recall the intertwining property $e^{-itH_\rho}W_\rho^\pm=W_\rho^\pm e^{-itH_{0,\rho}}$. We can then compute
\begin{eqnarray}
\lefteqn{i(S_\rho-1)\Phi_v=i(W_\rho^+-W_\rho^-)^*W_\rho^-\Phi_v}\nonumber\\
\Eqn{=}\int_{-\infty}^\infty e^{itH_{0,\rho}}V(x)e^{-itH_\rho}W_\rho^-\Phi_vdt=\int_{-\infty}^\infty e^{itH_{0,\rho}}V(x)W_\rho^-e^{-itH_{0,\rho}}\Phi_vdt
\end{eqnarray}
and
\begin{eqnarray}
|v|^{2\rho-1}(i(S_\rho-1)\Phi_v,\Psi_v)\Eqn{=}|v|^{2\rho-1}\int_{-\infty}^\infty(V(x)W_\rho^-e^{-itH_{0,\rho}}\Phi_v,e^{-itH_{0,\rho}}\Psi_v)dt\nonumber\\
\Eqn{=}|v|^{2\rho-1}\int_{-\infty}^\infty I_v(t)dt+R_v,\label{eq3.3.1}
\end{eqnarray}
where we defined $I_v(t)$ and $R_v$ in \eqref{eq3.3.1} by
\begin{gather}
I_v(t)=(V(x)e^{-itH_{0,\rho}}\Phi_v,e^{-itH_{0,\rho}}\Psi_v),\\
R_v=|v|^{2\rho-1}\int_{-\infty}^\infty((W_\rho^--1)e^{-itH_{0,\rho}}\Phi_v,V(x)e^{-itH_{0,\rho}}\Psi_v)dt.
\end{gather}
Proposition \ref{prop} and Corollary \ref{cor} immediately give
\begin{equation}
R_v=O(|v|^{1-2\rho}).\label{eq3.3.2}
\end{equation}
Thus far, the proof has been roughly parallel to that in Enss--Weder \cite{EnWe}. However, the principal part of \eqref{eq3.3.1} demands further rigorous scrutiny. We first divide the integral as follows
\begin{equation}
|v|^{2\rho-1}\int_{-\infty}^\infty I_v(t)dt=|v|^{2\rho-1}\left(\int_{|t|<|v|^{-\sigma}}+\int_{|t|\geqslant|v|^{-\sigma}}\right)I_v(t)dt,
\end{equation}
where $\sigma>2\rho-1$ is independent of $t$ and $v$. We will later determine an upper bound on $\sigma$. Because $I_v(t)$ is uniformly bounded in $t$ and $v$, the integral on $|t|<|v|^{-\sigma}$ is
\begin{equation}
|v|^{2\rho-1}\int_{|t|<|v|^{-\sigma}}|I_v(t)|dt\leqslant C|v|^{2\rho-1-\sigma}.\label{eq3.3.3}
\end{equation}
We next consider the integral on $|t|\geqslant|v|^{-\sigma}$, which is represented by
\begin{eqnarray}
\lefteqn{|v|^{2\rho-1}\int_{|t|\geqslant|v|^{-\sigma}}I_v(t)dt=|v|^{2\rho-1}\int_{|t|\geqslant|v|^{-\sigma}}(V(x+(\nabla_\xi\omega_\rho)(v)t)\Phi_0,\Psi_0)dt}\qquad\nonumber\\
\Eqn{}+|v|^{2\rho-1}\int_{|t|\geqslant|v|^{-\sigma}}\left\{I_v(t)-(V(x+(\nabla_\xi\omega_\rho)(v)t)\Phi_0,\Psi_0)\right\}dt.\label{eq3.3.4}
\end{eqnarray}
We note that $(\nabla_\xi\omega_\rho)(v)=|v|^{2\rho-2}v$. After the change of the integral variable $\tau=|v|^{2\rho-1}t$, the first term of the right-hand side of \eqref{eq3.3.4} converges
\begin{eqnarray}
\lefteqn{|v|^{2\rho-1}\int_{|t|\geqslant|v|^{-\sigma}}(V(x+(\nabla_\xi\omega_\rho)(v)t)\Phi_0,\Psi_0)dt}\nonumber\\
\Eqn{=}\int_{|\tau|\geqslant|v|^{2\rho-1-\sigma}}(V(x+\hat{v}\tau)\Phi_0,\Psi_0)d\tau\longrightarrow\int_{-\infty}^\infty(V(x+\hat{v}\tau)\Phi_0,\Psi_0)d\tau
\end{eqnarray}
as $|v|\rightarrow\infty$ because we assumed that $2\rho-1-\sigma<0$. This also indicates that
\begin{eqnarray}
\lefteqn{|v|^{2\rho-1}\int_{|t|\geqslant|v|^{-\sigma}}(V(x+(\nabla_\xi\omega_\rho)(v)t)\Phi_0,\Psi_0)dt}\nonumber\\
\Eqn{=}\int_{-\infty}^\infty(V(x+\hat{v}t)\Phi_0,\Psi_0)dt+O(|v|^{2\rho-1-\sigma})\label{eq3.3.5}
\end{eqnarray}
by the uniformly boundedness of $(V(x+\hat{v}t)\Phi_0,\Psi_0)$. Recall the relations \eqref{eq1.1.1} and \eqref{eq1.1.2}. We then have
\begin{equation}
I_v(t)=(V(x+(\nabla_\xi\omega_\rho)(D_x+v)t)\Phi_0,\Psi_0).
\end{equation}
Therefore, as in the proof of Theorem \ref{the1}, we try to derive the order of decay in the second term on the right-hand side of \eqref{eq3.3.4} from the nearly cancellation of $(\nabla_\xi\omega_\rho)(\xi+v)$ and $(\nabla_\xi\omega_\rho)(v)$ on the support of $\mathscr{F}\Phi_0$. In our assumptions, $V$ belongs to $C^1(\mathbb{R}^n)$, however we can compute
\begin{eqnarray}
\lefteqn{V\Bigl(x+(\nabla_\xi\omega_\rho)(\xi+v)t\Bigr)-V\Bigl(x+(\nabla_\xi\omega_\rho)(v)t\Bigr)}\nonumber\\
\Eqn{=}\int_0^1(\nabla_xV)\Bigl(x+(\nabla_\xi\omega_\rho)(v)t+\theta\{(\nabla_\xi\omega_\rho)(\xi+v)-(\nabla_\xi\omega_\rho)(v)\}t\Bigr)\nonumber\\
\Eqn{}\hspace{50mm}\cdot\{(\nabla_\xi\omega_\rho)(\xi+v)-(\nabla_\xi\omega_\rho)(v)\}td\theta\label{eq3.3.6}
\end{eqnarray}
as the pseudo-differential symbolic calculus in the Fourier integral. We particularly note that the second- and higher-order derivatives of $V$ do not appear on the right-hand side of \eqref{eq3.3.6} because $(\nabla_\xi\omega_\rho)(\xi+v)-(\nabla_\xi\omega_\rho)(v)$ does not include $x$. Let $f_1,f_2\in C_0^\infty(\mathbb{R}^n)$ satisfy $\mathscr{F}\Phi_0=f_1\mathscr{F}\Phi_0$ and $f_1=f_2f_1$. Then $\Phi_0=f_2(D_x)f_1(D_x)\Phi_0$ holds. We define $g_{j,v}$ by
\begin{equation}
g_{j,v}(\xi)=\{(\partial_{\xi_j}\omega_\rho)(\xi+v)-(\partial_{\xi_j}\omega_\rho)(v)\}f_1(\xi)
\end{equation}
for $1\leqslant j\leqslant n$ and, as in \eqref{eq1.1.a}, \eqref{eq1.1.b} and \eqref{eq1.1.c}
\begin{equation}
|\partial_\xi^\beta g_{j,v}(\xi)|\leqslant C_\beta|v|^{2\rho-2}\label{eq3.3.7}
\end{equation}
follows for any $\beta$. We also define the vector-valued function $\psi_v$ by
\begin{equation}
\psi_v(t,x,\xi)=x+(\nabla_\xi\omega_\rho)(v)t+\theta\{(\nabla_\xi\omega_\rho)(\xi+v)-(\nabla_\xi\omega_\rho)(v)\}t
\end{equation}
to avoid complicated notation. Now, to estimate the second term in \eqref{eq3.3.4}, we only have to consider the following norm, which includes the integrand of the $j$-th term on the right-hand side of \eqref{eq3.3.6},
\begin{equation}
|t|\|(\partial_{x_j}V)(\psi_v(t,x,D_x))f_2(D_x)g_{j,v}(D_x)\Phi_0\|\leqslant J_1+J_2,
\end{equation}
where $J_1$ and $J_2$ are given by
\begin{eqnarray}
J_1\Eqn{=}|t|\left\|(\partial_{x_j}V)(\psi_v(t,x,D_x))f_2(D_x)\chi\left(\frac{x}{|v|^{2\rho-1}|t|/4}\right)g_{j,v}(D_x)\Phi_0\right\|,\nonumber\\
J_2\Eqn{=}|t|\left\|(\partial_{x_j}V)(\psi_v(t,x,D_x))f_2(D_x)\left\{1-\chi\left(\frac{x}{|v|^{2\rho-1}|t|/4}\right)\right\}g_{j,v}(D_x)\Phi_0\right\|.\qquad
\end{eqnarray}
For $J_1$, we insert
\begin{equation}
F\left(|x|>\frac{|v|^{2\rho-1}|t|}{4}\right)+F\left(|x|\leqslant\frac{|v|^{2\rho-1}|t|}{4}\right)=1
\end{equation}
between $g_{j,v}(D_x)$ and $\Phi_0$. Then $J_1$ is estimated so that $J_1\leqslant J_{1,1}+J_{1,2}$, where $J_{1,1}$ and $J_{1,2}$ are given by
\begin{eqnarray}
J_{1,1}\Eqn{=}C|t|\left\|F\left(|x|>\frac{|v|^{2\rho-1}|t|}{4}\right)\Phi_0\right\|,\\
J_{1,2}\Eqn{=}C|t|\left\|\chi\left(\frac{x}{|v|^{2\rho-1}|t|/4}\right)g_{j,v}(D_x)F\left(|x|\leqslant\frac{|v|^{2\rho-1}|t|}{4}\right)\Phi_0\right\|.
\end{eqnarray}
The estimate of $J_{1,1}$ is almost the same as \eqref{eq3.2.5}. However, in this estimate, we choose $\nu\in\mathbb{R}$ such that
\begin{equation}
J_{1,1}\leqslant C|t|\left\|F\left(|x|>\frac{|v|^{2\rho-1}|t|}{4}\right)\langle x\rangle^{-\nu}\right\|\|\langle x\rangle^\nu\Phi_0\|\leqslant C|t|\langle|v|^{2\rho-1}|t|\rangle^{-\nu}.
\end{equation}
Therefore, for $\nu>2$, we obtain
\begin{eqnarray}
\lefteqn{|v|^{2\rho-1}\int_{|t|\geqslant|v|^{-\sigma}}J_{1,1}dt\leqslant C|v|^{2\rho-1}\int_{|t|\geqslant|v|^{-\sigma}}|t|\langle|v|^{2\rho-1}|t|\rangle^{-\nu}dt}\nonumber\\
\Eqn{}\leqslant C|v|^{-(2\rho-1)(\nu-1)}\int_{|v|^{-\sigma}}^\infty t^{-\nu+1}dt=O(|v|^{-(2\rho-1)(\nu-1)+\sigma(\nu-2)}).\label{eq3.3.8}
\end{eqnarray}
Although this estimate holds for any $\nu>2$, the exponent is better when $\nu$ is closer to $2$ because
\begin{equation}
-(2\rho-1)(\nu-1)+\sigma(\nu-2)=(\nu-2)\{\sigma-(2\rho-1)\}+1-2\rho\label{eq3.3.8sub}
\end{equation}
and $\sigma>2\rho-1$. In the estimate of $J_{1,2}$, we compute the following commutator by using the pseudo-differential product formula \eqref{product_formula}
\begin{eqnarray}
\lefteqn{\left[\chi\left(\frac{x}{|v|^{2\rho-1}|t|/4}\right),g_{j,v}(\xi)\right]}\nonumber\\
\Eqn{=}-\sum_{1\leqslant|\beta|\leqslant N-1}\frac{1}{\beta!}\partial_\xi^\beta g_{j,v}(\xi)\times(-i\partial_x)^\beta\chi\left(\frac{x}{|v|^{2\rho-1}|t|/4}\right)+R_N(t,x,\xi)
\end{eqnarray}
for any $N\in\mathbb{N}$. As in the proof of Theorem \ref{the1}, the disjointness of two characteristic functions means that, for $0\leqslant|\beta|\leqslant N-1$,
\begin{equation}
\left\{\partial_x^\beta\chi\left(\frac{x}{|v|^{2\rho-1}|t|/4}\right)\right\}F\left(|x|\leqslant\frac{|v|^{2\rho-1}|t|}{4}\right)=0.
\end{equation}
Therefore, $J_{1,2}$ only has the remainder term $R_N$. To estimate $R_N$, we divide the integral again
\begin{equation}
|v|^{2\rho-1}\int_{|t|\geqslant|v|^{-\sigma}}J_{1,2}dt=|v|^{2\rho-1}\left(\int_{|v|^{-\sigma}\leqslant|t|<|v|^{1-2\rho}}+\int_{|t|\geqslant|v|^{1-2\rho}}\right)J_{1,2}dt.
\end{equation}
By using the $L^2$-boundedness \eqref{bounds}, when $|t|\geqslant|v|^{1-2\rho}$, $R_N(t,x,D_x)$ is estimated as
\begin{equation}
\|R_N(t,x,D_x)\|\leqslant C_N|v|^{2\rho-2}(|v|^{2\rho-1}|t|)^{-N}
\end{equation}
because $|v|^{2\rho-1}|t|\geqslant1$ holds in this case, where $|v|^{2\rho-2}$ comes from \eqref{eq3.3.7}. We therefore compute, for $N\geqslant3$
\begin{eqnarray}
\lefteqn{|v|^{2\rho-1}\int_{|t|\geqslant|v|^{1-2\rho}}J_{1,2}dt=C|v|^{2\rho-1}\int_{|t|\geqslant|v|^{1-2\rho}}|t|\|R_N(t,x,D_x)\|dt}\nonumber\\
\Eqn{}\qquad\leqslant C_N|v|^{2(2\rho-1)-1-(2\rho-1)N}\int_{|v|^{1-2\rho}}^\infty t^{-N+1}dt=O(|v|^{-1}).\qquad\label{eq3.3.9}
\end{eqnarray}
In contrast, when $|v|^{-\sigma}\leqslant|t|<|v|^{1-2\rho}$, there exists $\tilde{N}\geqslant N$ such that
\begin{equation}
\|R_N(t,x,D_x)\|\leqslant C_N|v|^{2\rho-2}(|v|^{2\rho-1}|t|)^{-\tilde{N}}
\end{equation}
because $|v|^{2\rho-1}|t|<1$ holds, and
\begin{eqnarray}
\lefteqn{|v|^{2\rho-1}\int_{|v|^{-\sigma}\leqslant|t|<|v|^{1-2\rho}}J_{1,2}dt=C|v|^{2\rho-1}\int_{|v|^{-\sigma}\leqslant|t|<|v|^{1-2\rho}}|t|\|R_N(t,x,D_x)\|dt}\nonumber\\
\Eqn{}\hspace{20mm}\leqslant C_N|v|^{2(2\rho-1)-1-(2\rho-1)\tilde{N}}\int_{|v|^{-\sigma}}^{|v|^{1-2\rho}} t^{-\tilde{N}+1}dt\nonumber\\
\Eqn{}\hspace{20mm}=O(|v|^{-1})+O(|v|^{2(2\rho-1)-1-(2\rho-1)\tilde{N}+\sigma(\tilde{N}-2)})\hspace{25mm}\label{eq3.3.10}
\end{eqnarray}
is obtained. This estimate holds for any $\tilde{N}\geqslant N$ ($\geqslant3$). However, the best exponent is the smallest $\tilde{N}$ because
\begin{equation}
2(2\rho-1)-1-(2\rho-1)\tilde{N}+\sigma(\tilde{N}-2)=(\tilde{N}-2)\{\sigma-(2\rho-1)\}-1.\label{eq3.3.10sub}
\end{equation}
From \eqref{eq3.3.8}, \eqref{eq3.3.8sub}, \eqref{eq3.3.9}, \eqref{eq3.3.10}, and \eqref{eq3.3.10sub}, we have
\begin{eqnarray}
|v|^{2\rho-1}\int_{|t|\geqslant|v|^{\sigma}}J_1dt=O(|v|^{(\nu-2)\{\sigma-(2\rho-1)\}+1-2\rho})+O(|v|^{(\tilde{N}-2)\{\sigma-(2\rho-1)\}-1}).\label{J_1}
\end{eqnarray}
We next consider $J_2$. On the supports of $f_2$ and $1-\chi$,
\begin{eqnarray}
\lefteqn{|\psi_v(t,x,\xi)|\geqslant|v|^{2\rho-1}|t|-|x|-|(\nabla_\xi\omega_\rho)(\xi+v)-(\nabla_\xi\omega_\rho)(v)||t|}\nonumber\\
\Eqn{}\geqslant|v|^{2\rho-1}|t|/2-|v|^{2\rho-1}|t|/8=3|v|^{2\rho-1}|t|/8\geqslant|v|^{2\rho-1}|t|/4\label{eq3.3.11}
\end{eqnarray}
holds for large $|v|$, here we used \eqref{eq1.1.4}. This says that
\begin{eqnarray}
\lefteqn{f_2(\xi)\left\{1-\chi\left(\frac{x}{|v|^{2\rho-1}|t|/4}\right)\right\}}\nonumber\\
\Eqn{=}\chi\left(\frac{\psi_v(t,x,\xi)}{|v|^{2\rho-1}|t|/8}\right)f_2(\xi)\left\{1-\chi\left(\frac{x}{|v|^{2\rho-1}|t|/4}\right)\right\},\label{eq3.3.11a}
\end{eqnarray}
because $\chi(\psi_v(t,x,\xi)/(|v|^{2\rho-1}|t|/8))=1$ by \eqref{eq3.3.11}. However, symbolically, \eqref{eq3.3.11a} is
\begin{eqnarray}
\lefteqn{\sum_{|\beta|\leqslant M-1}\frac{1}{\beta!}\left\{\partial_\xi^\beta\chi\left(\frac{\psi_v(t,x,\xi)}{|v|^{2\rho-1}|t|/8}\right)\right\}}\nonumber\\
\Eqn{}\hspace{10mm}\times f_2(\xi)(-i\partial_x)^\beta\left\{1-\chi\left(\frac{x}{|v|^{2\rho-1}|t|/4}\right)\right\}+R_M(t,x,\xi)
\end{eqnarray}
for any $M\in\mathbb{N}$ by using the asymptotic product formula \eqref{product_formula} again. We note that
\begin{equation}
\left|\partial_\xi^\beta\chi\left(\frac{\psi_v(t,x,\xi)}{|v|^{2\rho-1}|t|/8}\right)\right||f_2(\xi)|\leqslant C_\beta|v|^{-(2\rho-1)}\leqslant C_\beta
\end{equation}
for any $\beta$ with $|\beta|\geqslant1$ and $C_\beta$ is independent of $t$. Therefore, for $0\leqslant|\beta|\leqslant M-1$, it is sufficient to consider
\begin{equation}
\chi\left(\frac{\psi_v(t,x,\xi)}{|v|^{2\rho-1}|t|/8}\right)f_2(\xi)\times\partial_x^\beta\left\{1-\chi\left(\frac{x}{|v|^{2\rho-1}|t|/4}\right)\right\}\label{eq3.3.12}
\end{equation}
only. The term which includes \eqref{eq3.3.12} is estimated to be
\begin{eqnarray}
\lefteqn{|t|\left\|(\partial_{x_j}V)(\psi_v(t,x,D_x))\chi\left(\frac{\psi_v(t,x,D_x)}{|v|^{2\rho-1}|t|/8}\right)\right\|}\nonumber\\
\Eqn{}\hspace{20mm}\times\left\|\partial_x^\beta\left\{1-\chi\left(\frac{x}{|v|^{2\rho-1}|t|/4}\right)\right\}\right\|\|g_{j,v}(D_x)\|\nonumber\\
\Eqn{}\leqslant C|t|\langle|v|^{2\rho-1}|t|\rangle^{-1-\gamma}(|v|^{2\rho-1}|t|)^{-|\beta|}|v|^{2\rho-2},
\end{eqnarray}
here we used the decay condition on $V$ in \eqref{potential} and the estimate of $g_{j,v}$ in \eqref{eq3.3.7}. We then compute the following integral
\begin{eqnarray}
\lefteqn{|v|^{2\rho-1}\int_{|t|\geqslant|v|^{-\sigma}}|t|\langle|v|^{2\rho-1}|t|\rangle^{-1-\gamma}(|v|^{2\rho-1}|t|)^{-|\beta|}|v|^{2\rho-2}dt}\nonumber\\
\Eqn{}\hspace{10mm}\leqslant C|v|^{2(2\rho-1)-1-(2\rho-1)(1+\gamma)-(2\rho-1)|\beta|}\int_{|v|^{-\sigma}}^\infty t^{-\gamma-|\beta|}dt\nonumber\\
\Eqn{}\hspace{10mm}=O(|v|^{2(2\rho-1)-1-(2\rho-1)(1+\gamma)-(2\rho-1)|\beta|+\sigma(\gamma+|\beta|-1)})\label{eq3.3.13}
\end{eqnarray}
because $\gamma>1$. The decay exponent in \eqref{eq3.3.13} is represented by
\begin{gather}
2(2\rho-1)-1-(2\rho-1)(1+\gamma)-(2\rho-1)|\beta|+\sigma(\gamma+|\beta|-1)\nonumber\\
=(\gamma+|\beta|-1)\{\sigma-(2\rho-1)\}-1.\label{eq3.3.13sub}
\end{gather}
Because $\sigma>2\rho-1$, the top term between $0\leqslant|\beta|\leqslant M-1$ is $|\beta|=M-1$. To estimate the term involving $R_M$, we have to divide the integral into $|v|^{-\sigma}\leqslant|t|<|v|^{1-2\rho}$ and $|t|\geqslant|v|^{1-2\rho}$ once more. By the same argument in \eqref{eq1.1.7}, there exists $M'\in\mathbb{N}$ such that
\begin{equation}
\|R_M(t,x,D_x)\|\leqslant C_M\sum_{0\leqslant j\leqslant M'}(|v|^{2\rho-1}|t|)^{-j}\times\sum_{M\leqslant j\leqslant M+M'}(|v|^{2\rho-1}|t|)^{-j},
\end{equation}
where, in the summation of $0\leqslant j\leqslant M'$, we used the following boundedness again
\begin{equation}
\left|\partial_\xi^\beta\chi\left(\frac{\psi_v(t,x,\xi)}{|v|^{2\rho-1}|t|/8}\right)\right||f_2(\xi)|\leqslant C_\beta
\end{equation}
for any $\beta$. Therefore, when $|t|\geqslant|v|^{1-2\rho}$, we have
\begin{equation}
\|R_M(t,x,D_x)\|\leqslant C_M(|v|^{2\rho-1}|t|)^{-M}\label{eq3.3.14}
\end{equation}
because $|v|^{2\rho-1}|t|\geqslant1$. On the other hand, in the case where $|v|^{-\sigma}\leqslant|t|<|v|^{1-2\rho}$,
\begin{equation}
\|R_M(t,x,D_x)\|\leqslant C_M(|v|^{2\rho-1}|t|)^{-\tilde{M}}\label{eq3.3.15}
\end{equation}
is obtained because $|v|^{2\rho-1}|t|<1$, here we put $\tilde{M}=M+2M'$. From \eqref{eq3.3.14}, \eqref{eq3.3.15}, and \eqref{eq3.3.7}, it follows that
\begin{eqnarray}
\lefteqn{|v|^{2\rho-1}\left(\int_{|v|^{-\sigma}\leqslant|t|<|v|^{1-2\rho}}+\int_{|t|\geqslant|v|^{1-2\rho}}\right)|t|\|R_M(t,x,D_x)\|\|g_{j,v}(D_x)\|dt}\nonumber\\
\Eqn{}\hspace{70mm}=O(|v|^{(\tilde{M}-2)\{\sigma-(2\rho-1)\}-1})\label{eq3.3.16}
\end{eqnarray}
for $M\geqslant3$. The computation in \eqref{eq3.3.16} is quite similar to \eqref{eq3.3.9} and \eqref{eq3.3.10} (see also \eqref{eq3.3.10sub}). By \eqref{eq3.3.13}, \eqref{eq3.3.13sub}, and \eqref{eq3.3.16}, $J_2$ is estimated to be
\begin{equation}
|v|^{2\rho-1}\int_{|t|\geqslant|v|^{-\sigma}}J_2dt=O(|v|^{(\gamma+1)\{\sigma-(2\rho-1)\}-1})+O(|v|^{(\tilde{M}-2)\{\sigma-(2\rho-1)\}-1}),\label{J_2}
\end{equation}
here we fixed $|\beta|=2$ in \eqref{eq3.3.13} and \eqref{eq3.3.13sub}, because we can choose $M=3$. By combining \eqref{eq3.3.2}, \eqref{eq3.3.3}, \eqref{eq3.3.5}, \eqref{J_1}, and \eqref{J_2}, we obtain
\begin{eqnarray}
\lefteqn{|v|^{2\rho-1}(i(S_\rho-1)\Phi_v,\Psi_v)=\int_{-\infty}^\infty(V(x+\hat{v}t)\Phi_0,\Psi_0)dt}\nonumber\\
\Eqn{}+\ O(|v|^{1-2\rho})+O(|v|^{2\rho-1-\sigma})\nonumber\\
\Eqn{}+\ O(|v|^{(\nu-2)\{\sigma-(2\rho-1)\}+1-2\rho})+O(|v|^{(\tilde{N}-2)\{\sigma-(2\rho-1)\}-1})\nonumber\\
\Eqn{}+\ O(|v|^{(\gamma+1)\{\sigma-(2\rho-1)\}-1})+O(|v|^{(\tilde{M}-2)\{\sigma-(2\rho-1)\}-1})
\end{eqnarray}
as $|v|\rightarrow\infty$. We evaluate these error exponents. It is clear that $2\rho-1-\sigma<0$ and, that $1-2\rho<(\nu-2)\{\sigma-(2\rho-1)\}+1-2\rho<0$ because we can choose $\nu-2>0$ to be sufficiently small, independent of the size of $\sigma$. Therefore, to complete this proof, we need to ensure $\sigma$ satisfies $(\tilde{N}-2)\{\sigma-(2\rho-1)\}-1<0$, $(\gamma+1)\{\sigma-(2\rho-1)\}-1<0$ and $(\tilde{M}-2)\{\sigma-(2\rho-1)\}-1<0$ on condition that $\sigma>2\rho-1$. To do that, it suffices to choose $\sigma$ such that
\begin{equation}
2\rho-1<\sigma<2\rho-1+\min\{1/(1+\gamma),1/(\tilde{N}-2),1/(\tilde{M}-2)\}
\end{equation}
for $\gamma>1$, $\tilde{N}\geqslant3$ and $\tilde{M}\geqslant3$. This completes the proof.
\end{proof}
From the Plancherel formula associated with the Radon transformation (see Helgason \cite{He}), the proof of Theorem \ref{the2} can be performed in the same way as in Theorem 1.1 of Enss--Weder \cite{EnWe}. We thus omit the proof here.\\

\noindent\textbf{Acknowledgments.} This work was partially supported by the Grant-in-Aid for Young Scientists (B) \#16K17633 from JSPS. Moreover, the author would like to thank the late Professor Hitoshi Kitada for many valuable discussions and comments.


\end{document}